\documentclass[11pt,a4]{article}
\pdfoutput=1

\usepackage[T1]{fontenc}
\usepackage{microtype}
\usepackage{fullpage}
\usepackage{authblk}
\usepackage{amsthm}
\usepackage{amssymb}
\usepackage[english]{babel}
\usepackage[utf8]{inputenc}
\usepackage{amsmath}
\usepackage{amsfonts}
\usepackage{graphicx}
\usepackage{mathtools}
\usepackage{enumerate}
\usepackage{microtype}
\usepackage{array}
\usepackage{url}

\newcolumntype{C}[1]{>{\centering\let\newline\\\arraybackslash\hspace{0pt}}m{#1}}

\newtheorem{definition}{Definition}[section]
\newtheorem{lemma}[definition]{Lemma}
\newtheorem{theorem}[definition]{Theorem}

\newtheorem{corollary}[definition]{Corollary}

\newtheorem{conjecture}[definition]{Conjecture}

\newcommand{\bigo}{\mathcal{O}}
\newcommand{\ham}{\mathrm{Ham}}
\newcommand{\runs}{\mathrm{runs}}
\newcommand{\enc}[2]{\big\{#1\big\}_{#2}}
\newcommand{\polylog}{\,\mathrm{polylog}\,}

\title{Optimal trade-offs for pattern matching\\ with $k$ mismatches}

\author[1]{Pawe\l{} Gawrychowski}
\affil[1]{Haifa University, Israel}
\author[2]{Przemys\l{}aw~Uzna\'nski}
\affil[2]{
ETH Z\"urich, Switzerland}
\date{}
\begin{document}
\maketitle

\setcounter{page}{1}

\begin{abstract}
Given a pattern of length $m$ and a text of length $n$, the goal in $k$-mismatch pattern matching
is to compute, for every $m$-substring of the text, the exact Hamming distance to the pattern
or report that it exceeds $k$. This can be solved in either $\widetilde\bigo(n \sqrt{k})$ time as shown
by Amir et al. [J. Algorithms 2004] or $\widetilde\bigo((m + k^2) \cdot n/m)$ time due to a result
of Clifford et al. [SODA 2016]. We provide a smooth time trade-off between these two bounds by
designing an algorithm working in time $\widetilde\bigo( (m + k \sqrt{m}) \cdot n/m)$.
We complement this with a matching conditional lower bound, showing that a significantly faster
\emph{combinatorial} algorithm is not possible, unless the combinatorial matrix multiplication
conjecture fails.
\end{abstract}



\section{Introduction}

The basic question in algorithms on strings is pattern matching, which asks for reporting (or detecting)
occurrences of the given pattern in the text. This fundamental question comes in multiple shapes
and colors, starting from the exact version considered already in the 70s~\cite{KMP}. Here we are
particularly interested in the approximate version, where the goal is to detect fragments of the text
that are \emph{similar} to the text. Two commonly considered variants of this question is pattern
matching with $k$ errors and pattern matching with $k$ mismatches. In the former, we are looking
for a fragment with edit distance at most $k$ to the pattern, while in the latter we are interested
in a fragment that differs from the pattern on up to $k$ positions (and has the same length).
The classical solution by Landau and Vishkin \cite{LandauV86} solves pattern matching with $k$ mismatches
in $\bigo(nk)$ time for a text of length $n$. For larger values of $k$, Abrahamson~\cite{Abrahamson87}
showed how to compute the number of mismatches between every fragments of the text and the
pattern of length $m$ in total $\bigo(n\sqrt{m\log m})$ time with convolution. Later, Amir et al.~\cite{AmirLP04}
combined both approaches to achieve $\bigo(n\sqrt{k\log k})$ time.

An obvious and intriguing question is what are the best possible time bounds for pattern matching
with $k$ mismatches. An unpublished result attributed to Indyk~\cite{Clifford} is that, if we are interested in
counting mismatches for every position in the text, then this is at least as difficult as multiplying
boolean matrices. In particular, it implies that one should not hope to significantly improve on the
$\bigo(n\sqrt{m})$ time complexity of an \emph{combinatorial} algorithm. However, this is not
sensitive to the bound $k$ on the number of mismatches. In a recent breakthrough, Clifford et
al.~\cite{CliffordFPSS16} introduced a new repertoire of tools and showed an
$\bigo((k^{2}\log k + m\polylog m)\cdot n/m)$ time algorithm.
In particular, this is near linear-time for $k=\bigo(\sqrt{m})$ and improves on the previous
algorithm of Amir et al.~\cite{AmirLP04} that runs in $\bigo(n/m \cdot (k^{3}\log k + m))$ time.

\paragraph{Results. }
 We provide a smooth transition between the $\widetilde\bigo(n\sqrt{k})$ time algorithm of
 Amir et al.~\cite{AmirLP04} and the $\widetilde\bigo((m+k^{2})\cdot n/m)$ solution given by
 Clifford et al.~\cite{CliffordFPSS16}. The running time of our algorithm is $\widetilde\bigo((m+k\sqrt{m})\cdot n/m)$.
 This matches the previous solution at the extreme points $k=\bigo(\sqrt{m})$ and $k=\Omega(m)$,
 but provides a better trade-off in-between. Furthermore, we prove that such transition
 is essentially the best possible. More precisely, we complement the algorithm with a matching
 conditional lower bound, showing that a significantly faster \emph{combinatorial} algorithm is not possible,
 unless the popular combinatorial matrix multiplication conjecture fails.
 
 \begin{figure}
\centering\includegraphics[width=0.5\textwidth]{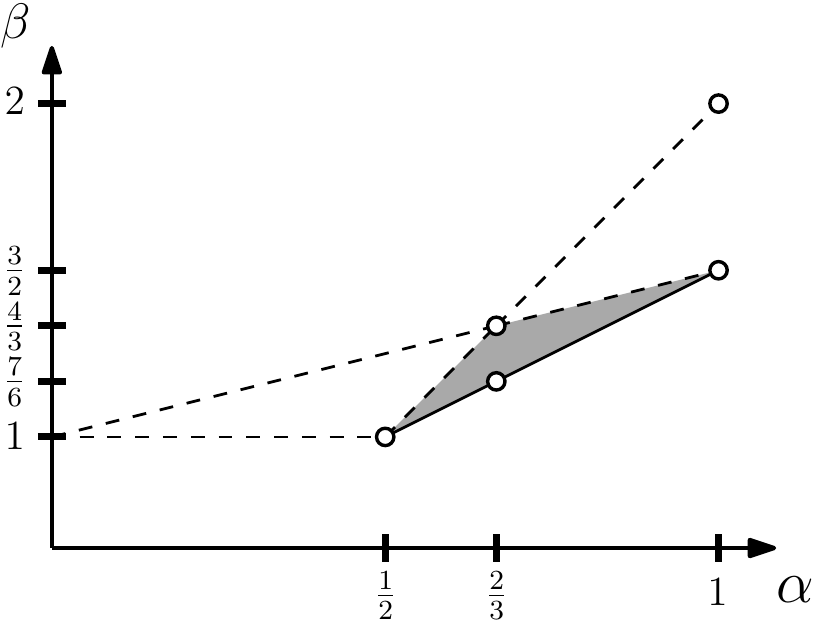}
\caption{Running time $T = m^\beta$ on instances with $n = \Theta(m)$ and $k = m^\alpha$.
Previous algorithms are represented by dashed lines and our algorithm is represented by solid line.
For example, for $k = \Theta(m^{2/3})$ we improve the complexity from $\widetilde\bigo(m^{4/3})$
to $\widetilde\bigo(m^{7/6})$. }
\end{figure}

\paragraph{Related work.}
Landau and Vishkin \cite{LandauV86} solve pattern matching with $k$ mismatches by checking
every possible alignment with $k+1$ constant-time longest common extension queries (also
known as ``kangaroo jumps''). The main idea in all the subsequent improvements is to use convolution,
which essentially counts matches generated by a particular letter with a single FFT in time close
to linear. Both Abrahamson \cite{Abrahamson87} and Amir et al.~\cite{AmirLP04} use convolution for letters often occurring in the
pattern. Convolution is also used (together with random projections $\Sigma \to \{0,1\}$ that can be
derandomized with an extra $\bigo(\log n)$ factor) by Karloff \cite{Karloff93} for approximate
mismatches counting.

At a very high level, Clifford et al.~\cite{CliffordFPSS16} obtain the improved time complexity by
partitioning both the pattern and the text into $\bigo(k)$ subpatterns and subtexts, such that the
total number of blocks in their RLE is small. Resulting $\bigo(k^2)$ instances of RLE pattern matching
with mismatches are then solved in $\bigo(k^2)$ total time, leading to an $\widetilde\bigo((k^2+ m)\cdot n/m)$
time algorithm for the original problem.

\paragraph{Overview of the techniques.}
We observe that the reduction from \cite{CliffordFPSS16} can be done so that, instead of many small instances,
we end up with a \emph{single} new instance of $\bigo(k)$-mismatch 
pattern matching. The resulting new pattern and text have RLE consisting of $\bigo(k)$ blocks
and the problem is reduced to RLE pattern matching with $k$ mismatches.
Since for RLE pattern matching with mismatches there is a matching quadratic conditional lower bound
(by reducing from the 3SUM problem), it might seem that no improvement here is possible without
making a significant breakthrough.

We show that this is not necessarily the case, by leveraging that the RLE strings are compressed version of strings 
of  $\bigo(m)$ length. Thus, letters that appear in only a few blocks of the compressed
pattern can be treated in a fashion similar to \cite{AmirLP04} by producing a representation of all matches
generated by a block in the compressed pattern against a block in the compressed text, in constant time
per a pair of blocks. For letters that appear in many blocks, we can essentially ``uncompress'' the corresponding
fragment of the pattern, and run the classical convolution, taking advantage of the fact that uncompressed 
versions are of length $\bigo(m)$. Setting threshold appropriately, we solve the obtained of RLE pattern
matching in time $\widetilde\bigo(k \sqrt{m})$ time. All in all, we obtain an
$\widetilde\bigo((m+k\sqrt{m})\cdot n/m)$ time solution to the original problem.

\section{Upper bound}

The goal of this section is to prove the following theorem:

\begin{theorem}
$k$-mismatch pattern matching can be solved in time $\bigo(n/m \cdot (m \log^2 m \log |\Sigma| + k \sqrt{m \log m}))$.
\end{theorem}

We begin with the standard trick of reducing the problem to $\lceil n/m \rceil$ instances of matching a
pattern $P$ of length $m$ to a text $T$ of length $2m$ and work with such formulation from now on.
Therefore, the goal now is to achieve $\bigo(m \log^2 m \log |\Sigma| + k \sqrt{m \log m}))$ complexity.

We start by highlighting the kernelization technique of Clifford et al.~\cite{CliffordFPSS16}.
An integer $\pi>0$ is an $x$-period of a string $S[1,m]$ if $\ham(S[\pi,m-1],S[0,m-1-\pi]) \le x$ (cf. Definition 1 in \cite{CliffordFPSS16}).
Note that compared to the original formulation, we drop the condition that $\pi$ is minimal
from the definition.

\begin{lemma}[Fact 3.1 in \cite{CliffordFPSS16}]
\label{lem:filtration}
If the minimal $2x$-period of the pattern is $\ell$, then the starting positions of any two occurrences
with $k$ mismatches of the pattern are at distance at least $\ell$.
\end{lemma}

The first step of algorithm is to determine the minimal $\bigo(k)$-period of the pattern. More specifically,
we run the $(1+\varepsilon)$-approximate algorithm of Karloff \cite{Karloff93} with $\varepsilon=1$
matching the pattern $P$ against itself. This takes $\bigo(m \log^2m \log |\Sigma|)$ time and, by looking at
the approximate outputs for offsets not larger than $k$, allows us to distinguish between two cases:
\begin{itemize}
\item every $2k$-period of the pattern is at least $k$, or
\item there is  a $4k$-period $\ell \le k$ of the pattern.
\end{itemize}
Then we run the appropriate algorithm as described below.

\paragraph{No small $2k$-period.}
We again run Karloff's algorithm with $\varepsilon=1$, but now we match the pattern with the text.
We look for positions $i$ where the approximate algorithm reports at most $k$ mismatches, meaning
that $\ham(P, T[i\ ..\ i+m-1]) \le 2k$. By Lemma~\ref{lem:filtration}, there are $\bigo(m/k)$ such positions,
and we can safely discard all other positions. Then, we test every such position using the ``kangaroo jumps''
technique of Landau and Vishkin~\cite{LandauV86}, using $\bigo(k)$ constant-time operations per position,
in total $\bigo(m)$ time.

\paragraph{Small $4k$-period.}
Let $\ell \le k$ be any $4k$-period of the pattern. 
For a string $S$ and $0 \le i < \ell$, let $\enc{S}{\ell,i} = S[i]S[i+\ell]S[i+2\ell]\ldots$ up until end of $S$.
We denote by $\enc{S}{\ell}$ an $\ell$-encoding of  $S$, that is the string 
$\enc{S}{\ell,1}\enc{S}{\ell,2}\ldots \enc{S}{\ell,\ell-1}$.
Let $\runs(S)$ be the number of runs in $S$. Denote $\runs_{\ell}(S) = \sum_{i=1}^{\ell} \runs(\enc{S}{\ell,i})$,
and observe that it upperbounds the number of runs in $\enc{S}{\ell}$.

\begin{lemma}[Lemma 6.1 in \cite{CliffordFPSS16}]
If $P$ has a $4k$-period not exceeding $k$, then $\runs_{\ell}(P) \le 5k$.
\end{lemma}

We proceed with the kernelization argument. Let $T_L$ be the longest suffix of $T[0,m-1]$
such that $\runs_{\ell}(T_L) \le 6k$. Similarly, let $T_R$ be the longest prefix of $T[m,2m-1]$
such that $\runs_{\ell}(T_R) \le 6k$. Let $T' = T_LT_R$. Obviously, $\runs_{\ell}(T') \le 12k$.

\begin{lemma}[Lemma 6.2 in \cite{CliffordFPSS16}]
Every $T[i,i+m-1]$ that is an occurrence of $P$ with $k$ mismatches is fully contained in $T'$.
\end{lemma}

Thus we see that $k$-mismatch pattern matching is reduced to a kernel where the $\ell$-encoding
of both the text and the pattern have few runs, that is, compress well with RLE.

From now on assume that both $T'$ and $P$ are of lengths divisible by $\ell$. If it is not the case, we can pad them separately with at most $\ell-1 < k$ characters each, not changing the complexity of our solution. 
Let $m_1$ and $m_2$ be integers such that $m_1 \cdot \ell = |T'|$ and $m_2 \cdot \ell = |P|$, $m_1 \ge m_2$.

\newcommand{\ddd}{\makebox[5.75pt]{\$}}
\begin{figure}[t]
\label{fig:rearrange}
\begin{minipage}{.5\textwidth}
$$T' = \mathtt{hokuspokusopensezame}$$
$$
\begin{array}{|C{6pt}|C{6pt}|C{6pt}|C{6pt}|C{6pt}|}
\hline
\texttt{h}&\texttt{s}&\texttt{u}&\texttt{e}&\texttt{z}\\
\hline
\texttt{o}&\texttt{p}&\texttt{s}&\texttt{n}&\texttt{a}\\
\hline
\texttt{k}&\texttt{o}&\texttt{o}&\texttt{s}&\texttt{m}\\
\hline
\texttt{u}&\texttt{k}&\texttt{p}&\texttt{e}&\texttt{e}\\
\hline
\end{array}
\quad
\begin{array}{|C{6pt}|C{6pt}|C{6pt}|C{6pt}|C{6pt}|}
\hline
\texttt{s}&\texttt{u}&\texttt{e}&\texttt{z}&\texttt{\#}\\
\hline
\texttt{p}&\texttt{s}&\texttt{n}&\texttt{a}&\texttt{\#}\\
\hline
\texttt{o}&\texttt{o}&\texttt{s}&\texttt{m}&\texttt{\#}\\
\hline
\texttt{k}&\texttt{p}&\texttt{e}&\texttt{e}&\texttt{\#}\\
\hline
\end{array}
$$
\end{minipage}
\begin{minipage}{.5\textwidth}
$$P = \mathtt{abracadabra}$$
$$\begin{array}{|C{6pt}|C{6pt}|C{6pt}|C{6pt}|C{6pt}|}
\hline
\texttt{a}&\texttt{c}&\texttt{b}&\texttt{\$}&\texttt{\$}\\
\hline
\texttt{b}&\texttt{a}&\texttt{r}&\texttt{\$}&\texttt{\$}\\
\hline
\texttt{r}&\texttt{d}&\texttt{a}&\texttt{\$}&\texttt{\$}\\
\hline
\texttt{a}&\texttt{a}&\texttt{\$}&\texttt{\$}&\texttt{\$}\\
\hline
\end{array}$$
\end{minipage}
\begin{align*}
T^\star\quad=&\quad\mathtt{hsuez\ opsna\ koosm\ ukpee\ suez\#\ psna\#\ oosm\#\ kpee\#}\\
P^\star\quad=&\quad\mathtt{acb\ddd\ddd\ bar\ddd\ddd\ rda\ddd\ddd\ aa\ddd\ddd\ddd}\\
\end{align*}
\caption{Example of rearranging of text and pattern, with parameter $\ell = 4$.}
\end{figure}


We rearrange both $P$ and $T'$ to take advantage of their regular structure. That is, we
define $T^\star = \enc{T'}{\ell}\enc{T''}{\ell}$, where $T'' = T'[ \ell+1, m_1 \cdot \ell]\ \#^\ell$. Observe that $T^\star$ is a word
of length $2 m_1 \cdot \ell$, composed first of $m_1$ blocks of the form $T'[i]T'[i+\ell]\ldots T'[i+(m_1-1)\ell]$
for $0 \le i < \ell$, and then of $m_1$ blocks of the form $T'[i+\ell]\ldots T'[i+(m_1-1)\ell]\ \#$.  

Similarly, we define $P^\star = \enc{P\ \$^{(m_1-m_2)\ell}}{\ell}$. Again we observe that $P^\star$ is the word of length $m_1\cdot \ell$, composed of blocks of the form
$P[i]P[i+\ell]\ldots P[i+(m_2-1)\ell]\ \$^{m_1-m_2}$ for $0 \le i < \ell$.  Example of this reduction is presented on Figure~\ref{fig:rearrange}.

Next we show that $T^\star$ and $P^\star$ maintain the Hamming distance between any possible alignment of $T'$ and $P$.

\begin{lemma}
For any integer $0 \le \alpha \le (m_1-m_2) \ell$, let $x = \lfloor \alpha / \ell \rfloor$ and $y = \alpha \bmod \ell$. Let $\beta = x + y \cdot m_1$. Then
$$\ham(T'[\alpha,\alpha + m_2\cdot \ell - 1],P)  = \ham(T^\star[\beta,\beta+m_1 \cdot \ell - 1], P^\star)-(m_1-m_2) \cdot \ell.$$
\end{lemma}
\begin{proof}
Observe that 
\begin{equation}
\label{eq:ham}
\ham(T'[\alpha,\alpha + m_2\cdot \ell - 1],P) = \sum_{i = 0}^{m_2  -1}  \sum_{j=0}^{\ell-1} \delta(T'[x \ell + y + i\ell+j], P[i\ell+j]), \end{equation}
where $\delta$ is indicator of character inequality.
Observe  that $P[i\ell+j]) = P^\star[i+j\cdot m_1]$, for $0 \le j < \ell - y$ ther is $T'[x \ell + y + i\ell+j] = T^\star[(x+i) + (y +j)m_1]$, and for $\ell - y \le j < \ell$ there is $T'[x \ell + y + i\ell+j] = T''[(x+i)\ell + (y+j-\ell)] = T^\star[(x+i) + (y+j-\ell)m_1 + \ell m_1] = T^\star[(x+i) + (y +j)m_1]$. Additionally, for $m_2 \le i < m_1$, $P^\star[i + j \cdot m_1] = \$$, which always generates a mismatch with any character in $T^\star$. Thus
\begin{align*}
\eqref{eq:ham} =& \sum_{i = 0}^{m_2 -1}  \sum_{j=0}^{\ell-1} \delta(T^\star[(x+i) + (y +j)m_1],P^\star[i + j \cdot m_1]) =\\
=& - (m_1-m_2)\ell + \sum_{i = 0}^{m_1 -1}  \sum_{j=0}^{\ell-1} \delta(T^\star[(x+i) + (y +j)m_1],P^\star[i + j \cdot m_1]), \tag*{\qedhere} 
\end{align*}
\end{proof}

We see that it is enough to find all occurrences of $P^\star$ in $T^\star$ with
$(k+(m_1-m_2)\cdot\ell)$ mismatches, where $k+(m_1-m_2)\ell \le 2k$,  $|P^\star| = |T'| \le m$ and $|T^\star| = 2|T'| \le 2m$. Additionally, $\runs(P^\star) \le 5k+\ell \le 6k$ and $\runs(T^\star) \le 12k + \ell \le 13k$.

Now we describe how to solve the kernelized problem exactly (where we count matches/mismatches for all possible alignments, not just detect occurrences with up to $k$ mismatches), using the stated properties of $T^\star$ and $P^\star$.

Consider a letter $c \in \Sigma$. For a string $S$, we denote by $\runs(S,c)$ the number of runs in $S$
consisting of occurrences of $c$.
Fix a parameter $t$. Call a letter $c$ such that $\runs(P^\star,c) > t$ a heavy letter, and otherwise call it
light. Now we describe how to count the number of mismatches for each type of letters. This is reminiscent
to a trick originally used by Abrahmson~\cite{Abrahamson87} and later refined by Amir et al.~\cite{AmirLP04}.

\paragraph{Heavy letters.}

For every heavy letter $c$ separately we use a convolution scheme. Since both $P^\star$ and $T^\star$ are of
size $\bigo(m)$, this takes time $\bigo(m \log m)$ per every such letter.
Since $\sum_{c \in \Sigma} \runs(P^\star,c) = \runs(P^\star) \le 6k$, there are $\bigo(k/t)$ heavy letters, making
the total time $\bigo(mk \log m / t)$.

\paragraph{Light letters.}
First, we preprocess $P^\star$, and for every light letter $c$ we compute a list of runs consisting of
occurrences of $c$. Our goal is to compute the array $A[0, |T^\star|-|P^\star|]$, where $A[i]$ counts
the number of matching occurrences of light letters in $T^\star[i,i+|P^\star|-1]$ and $P^\star$.

\begin{figure}[t]
\begin{minipage}{.5\textwidth}

\centering\includegraphics[width=0.75\textwidth]{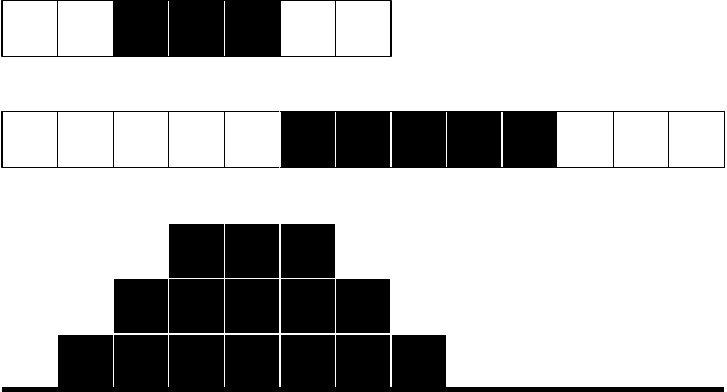}

\end{minipage}
\begin{minipage}{.5\textwidth}
\vspace{-0.5cm}
\begin{align*}
DA &= \begin{array}{|C{6pt}|C{6pt}|C{6pt}|C{6pt}|}
\hline
1 & 0 & 0 &\!-1\!\\
\hline
\end{array}\\
\\
DB &= \begin{array}{|C{6pt}|C{6pt}|C{6pt}|C{6pt}|C{6pt}|C{6pt}|}
\hline
1 & 0 & 0 & 0 & 0 &\!-1\!\\
\hline
\end{array}\\
\\
D^2 (A\cdot B) &= \begin{array}{|C{6pt}|C{6pt}|C{6pt}|C{6pt}|C{6pt}|C{6pt}|C{6pt}|C{6pt}|C{6pt}|}
\hline
1 & 0 & 0 &\!-1\!& 0 & \!-1\!& 0 & 0 & 1\\
\hline
\end{array}
\end{align*}
\end{minipage}
\caption{On the left - a run in the pattern and a run in the text (both represented by black boxes) consisting
of the same character and a histogram of the matches they generate.
On the right - first derivates of the indicator arrays and second derivate of the match array, without padding zeroes.}
\end{figure}

We scan $T^\star$, and for every run of a particular light letter, we iterate through the
precomputed list of runs of this letter in $P^\star$. Observe that, given a run of the same letter in $P^\star$ and in $T^\star$, matches generated between $T^\star[u,v]$ and $P^\star[y,z]$ account for a piecewise linear function. More precisely, for all integer $u \le i \le v$ and $y \le j \le z$, we need to increase $A[i-j]$ by one.
To see that we can process pair of runs in constant time, we work with discrete derivates, instead of original arrays.

Given sequence $F$, we define its discrete derivate $DF$ as follow: $(DF)[i] = F[i]-F[i-1]$. Correspondingly, if we consider generating function $F(x) = \sum_i F[i] x^i$, then $(DF)(x) = F(x) \cdot (1-x)$ (for convenience, we assume that arrays are indexed from $-\infty$ to $\infty$). 

Now consider indicator sequences $T_{u,v}[i] = \mathbf{1}( u \le i \le v )$ and $P_{y,z}[j] = \mathbf{1}( -z \le j \le -y )$. To perform the update, we set $A[i+j]\ +\!\!= T_{u,v}[i] \cdot P_{y,z}[j]$ for all $i,j$, or simpler using generating functions:
\begin{equation}
\label{eq:update}
A(x)\ +\!\!= T_{u,v}(x) \cdot P_{y,z}(x),
\end{equation}
where $T_{u,v}(x) = \sum_{i=u}^{v} x^i$ and $P_{y,z}(x) = \sum_{j=y}^z x^{-j}$. However, we observe that $D T_{u,v}$ and $D P_{y,z}$ have particularly simple forms: $D T_{u,v}(x) = x^u - x^{v+1}$ and $D P_{y,z}(x) = x^{-z} - x^{-y+1}$. Thus it is easier to maintain second derivate of $A$, and \eqref{eq:update} becomes:
$$
D^2 A(x)\ +\!\!= x^{u-z} - x^{v-z+1} - x^{u-y+1} + x^{v-y+2}.
$$

All in all, we can maintain $D^2A$ in constant time per pair of runs, or in $\bigo(k \cdot t)$ total time, since every list of runs is of length at most $t$,
and there are at most $13k$ runs in $T^\star$. Additionally, in $\bigo(m)$ time we can compute
$A[0]$ and $A[1]$, allowing us to recover all other $A[i]$s from the formula $A[i] = (D^2A)[i] + 2A[i-1]-A[i-2]$.

Setting $t = \sqrt{m \log m}$ gives the total running time $\bigo(k \sqrt{m \log m})$ in both cases
as claimed.

\section{Lower bound}
Below we present a conditional lower bound, which expands upon an idea attributed to Indyk~\cite{Clifford}.
Main idea here is to use rectangular matrices instead of square, and use the padding accordingly. However, we pad using the same character in both text and pattern, increasing the number of mismatches only by a factor of 2.

Recall the combinatorial matrix multiplication conjecture stating that, for any $\varepsilon > 0$,
there is no \emph{combinatorial} algorithm for multiplying two $n\times n$ boolean matrices working
in time $\bigo(n^{3-\varepsilon})$. The following formulation is equivalent to this conjecture:
\begin{conjecture}[Combinatorial matrix multiplication]
\label{conj:cmm}
For any $\alpha,\beta,\gamma,\varepsilon > 0$, there is no combinatorial algorithm for multiplying
an $n^{\alpha} \times n^{\beta}$ matrix with an $n^{\beta} \times n^{\gamma}$ matrix in time
$\bigo(n^{\alpha+\beta+\gamma - \varepsilon})$.
\end{conjecture}
The equivalence can be seen by simply cutting the matrices into square block (in one direction)
or in rectangular blocks (in the other direction).

Now, consider two boolean matrices, $A$ of dimension $M' \times N$ and $B$ of dimension $N \times M$, for $ M' \ge M \ge N$. We encode $A$ as text $T$, by encoding elements row by row and adding some padding. Namely:
\[T =\#^{M^2} r_1 \ \#^{M-N+1}\ r_2\ \#^{M-N+1}\ \ldots\ \#^{M-N+1}\ r_{M'} \#^{M^2}\]
where $r_i = r_{i,1}\ldots r_{i,N}$ and $r_{i,j} = 0$ when $A_{i,j} = 0$ and $r_{i,j} = j$ when $A_{i,j} = 1$.
Similarly, we encode $B$ as $P$ column by column, using padding shorter by one character:
\[P = c_1 \ \#^{M-N}\ c_2\ \#^{M-N}\ \ldots\ \#^{M-N}\ c_M\]
where $c_j = c_{1,j}\ldots c_{N,j}$ and $c_{i,j} = 0'$ when $B_{i,j} = 0$ and $c_{i,j} = i$ when $B_{i,j} = 1$.

Observe that, since we encode $0$s from $A$ and $B$ using different symbols, and encoding of $1$s
is position-dependent, $r_i$ and $c_j$ will generate a match only if they are perfectly aligned and 
there is $k$ such that $r_{i,k} = c_{k,j}$, or equivalently $A_{i,k} = B_{k,j} = 1$. Since each block
(encoded row plus following padding) is either of length $N+1$ for rows or $N$ for columns, there
will be at most one pair row-column aligned for each pattern-text alignment. 

The total number of mismatches, for each alignment, is at most $2NM$ (since there are at most $MN$ non-$\#$ text characters that are aligned with pattern, and at most $MN$ non-$\#$ pattern characters). We can recover whether any given entry of
$A \cdot B$ is a $1$, since if so the number of mismatches for the corresponding alignment is
decreased by 1.

We have $|T| = \Theta(M' M)$ and $|P| = \Theta(M^2)$. By setting $M = \sqrt{m}$, $M' = \frac{n}{\sqrt{m}}$ and $N = \frac{k}{\sqrt{m}}$ we have the following: 

\begin{corollary}
For any positive $\varepsilon,\alpha,\kappa$, such that $\frac{1}{2}\alpha \le \kappa \le \alpha \le 1$
there is no \emph{combinatorial} algorithm solving pattern matching with $k=\Theta(n^\kappa)$
mismatches in time $\bigo((k\sqrt{m} \cdot n/m)^{1-\varepsilon})$ for a text of length $n$ and
a pattern of length $m=\Theta(n^\alpha)$, unless Conjecture~\ref{conj:cmm} fails.
\end{corollary}

If we denote by $\omega(\alpha,\beta,\gamma)$ the exponent of fastest algorithm to multiply
a matrix of dimension $n^{\alpha} \times n^{\beta}$ with a matrix of dimension
$n^{\beta} \times n^{\gamma}$, we have:

\begin{corollary}
For any positive $\varepsilon,\alpha,\kappa$, such that $\frac{1}{2}\alpha \le \kappa \le \alpha \le 1$
there is no algorithm solving pattern matching with $\Theta(n^\kappa)$ mismatches in time
$\bigo(n^{\omega(2-\alpha,2\kappa-\alpha,\alpha)/2-\varepsilon})$ for a text of length $n$ and a
pattern of length $\Theta(n^\alpha)$.
\end{corollary}

\bibliographystyle{plain}
\bibliography{bib}


\end{document}